\documentclass[fleqn,11pt]{eptcs}
\usepackage{breakurl}                 

\usepackage[fleqn]{amsmath}
\usepackage{amssymb}
\usepackage{amsthm}
\usepackage{color}
\usepackage{eucal}
\usepackage{graphicx}
\usepackage{graphs}
\usepackage{procalg}
\usepackage{url}

\theoremstyle{definition}
\newtheorem{theorem}{Theorem}[section]
\newtheorem{corollary}[theorem]{Corollary}
\newtheorem{lemma}[theorem]{Lemma}
\newtheorem{proposition}[theorem]{Proposition}

\newtheorem{definition}[theorem]{Definition}
\newtheorem{example}[theorem]{Example}

\newcommand{\ACP}[1][\Act,\gamma]{\ensuremath{\mathrm{ACP}(#1)}}
\newcommand{\ACPt}[1][\Act,\gamma]{\ensuremath{\mathrm{ACP}_{\tau}(#1)}}
\newcommand{\ACPs}[1][\Act,\gamma]{\ensuremath{\mathrm{ACP}^{*}(#1)}}

\newcommand{\ACPsz}[1][\Act,\gamma]{\ensuremath{\mathrm{ACP}_{\dl}^{*}(#1)}}
\newcommand{\ACPszo}[1][\Act,\gamma]{\ensuremath{\mathrm{ACP}_{\dl,\emp}^{*}(#1)}}
\newcommand{\ACPszt}[1][\Act,\gamma]{\ensuremath{\mathrm{ACP}_{\dl,\tau}^{*}(#1)}}
\newcommand{\ACPszot}[1][\Act,\gamma]{\ensuremath{\mathrm{ACP}_{\dl,\emp,\tau}^{*}(#1)}}
\newcommand{\BPA}[1][\Act]{\ensuremath{\mathrm{BPA}(#1)}}
\newcommand{\BPAd}[1][\Act]{\ensuremath{\mathrm{BPA}_{\delta}(#1)}}
\newcommand{\BPAsd}[1][\Act]{\ensuremath{\mathrm{BPA}^{*}_{\delta}(#1)}}

\newcommand{\BPAsz}[1][\Act]{\ensuremath{\mathrm{BPA}^{*}_{\dl}(#1)}}
\newcommand{\BPAszo}[1][\Act]{\ensuremath{\mathrm{BPA}^{*}_{\dl,\emp}(#1)}}
\newcommand{\PA}[1][\Act]{\ensuremath{\mathrm{PA}(#1)}}
\newcommand{\PAd}[1][\Act]{\ensuremath{\mathrm{PA}_{\delta}(#1)}}
\newcommand{\PAsd}[1][\Act]{\ensuremath{\mathrm{PA}^{*}_{\delta}(#1)}}

\newcommand{\PAsz}[1][\Act]{\ensuremath{\mathrm{PA}^{*}_{\dl}(#1)}}
\newcommand{\PAszo}[1][\Act]{\ensuremath{\mathrm{PA}^{*}_{\dl,\emp}(#1)}}
\newcommand{\OC}[1]{\ensuremath{\#(#1)}}
\newcommand{\lexp}{\prec}

\renewcommand{\max}[1]{\mathrm{max}\{#1\}}

\newcommand{\notstep}[1]{\cnot{\step{#1}}}
\newcommand{\notmstep}[1]{\mathrel{{\cnot{\step{}}}^{+}}}
\newcommand{\notsteps}[1]{\mathrel{{\cnot{\step{}}}^{*}}}
\newcommand{\PEXP}[1][]{\ensuremath{\mathcal{P}_{#1}}}
  \newcommand{\ACPszoExp}{\PEXP[\ACPszo]}
  \newcommand{\PAszoExp}{\PEXP[\PAszo]}
  \newcommand{\BPAszoExp}{\PEXP[\BPAszo]}
\newcommand{\stepsym}[1][]{\ensuremath{\mathalpha{\rightarrow_{#1}}}}
\renewcommand{\step}[2][]{\ensuremath{\mathbin{\arrow{#2}_{#1}}}}
  \newdimen\boxwdplusemdimen
  \def\arrow#1{{
     \boxwdplusemdimen=1em%
     \setbox0=\hbox{$\scriptstyle#1$}%
     \advance\boxwdplusemdimen by \wd0\relax%
     \ifdim\boxwdplusemdimen<16.11119pt%
       \boxwdplusemdimen=16.11119pt%
     \fi%
     \buildrel{#1}\over%
       {\setbox1=\hbox to \boxwdplusemdimen{\rightarrowfill}%
     \ht1=0.3em\relax\box1}%
   }}
\renewcommand{\term}[2][]{\ensuremath{{#2}\mathclose{\downarrow_{#1}}}}
\newcommand{\termsym}[1][]{\ensuremath{\mathalpha{\downarrow_{#1}}}}
\newcommand{\notterm}[2][]{\ensuremath{{#2}\mathclose{{\cnot\downarrow}_{#1}}}}
\newcommand{\brelsym}[1][]{\ensuremath{\mathalpha{\mathcal{R}_{#1}}}}
\newcommand{\brel}[1][]{\ensuremath{\mathrel{\brelsym[#1]}}}
\newcommand{\TSS}[1][]{\ensuremath{\mathcal{T}_{#1}}}
  \newcommand{\ACPszoTSS}{\TSS[\ACPszo]}
  \newcommand{\PAszoTSS}{\TSS[\PAszo]}
  \newcommand{\BPAszoTSS}{\TSS[\BPAszo]}
\newcommand{\FA}[1][]{\ensuremath{\mathcal{F}_{#1}}}
\newcommand{\Ext}[1]{\ensuremath{ET(#1)}}
\newcommand{\Extn}[1]{\ensuremath{ET_{n}(#1)}}

\newcommand{\Reach}[1]{\ensuremath{[#1]_{\rightarrow}}}
\newcommand{\steprtc}[1][]{\ensuremath{\mathbin{\step[#1]{}^{*}}}}
\newcommand{\steptc}[1][]{\ensuremath{\mathbin{\step[#1]{}^{+}}}}
\newcommand{\nfa}{finite automaton}
\newcommand{\nfas}{finite automata}

\newcommand{\Lemma}[1]{Lemma~\ref{lem:#1}}

\title{Expressiveness modulo Bisimilarity of Regular Expressions
  with Parallel Composition\\
  \normalfont (Extended Abstract)
}

\author{%
  Jos C. M. Baeten
    \institute{Eindhoven University of Technology, The Netherlands}
    \email{j.c.m.baeten@tue.nl}
\and
  Bas Luttik
    \institute{Eindhoven University of Technology, The Netherlands}
    \institute{Vrije Universiteit Amsterdam, The Netherlands}
    \email{s.p.luttik@tue.nl}
\and
  Tim Muller
    \institute{\hspace*{22pt}University of Luxembourg, Luxembourg\hspace*{22pt}}
    \email{tim.muller@uni.lu}
\and
  Paul van Tilburg
    \institute{Eindhoven University of Technology, The Netherlands}
    \email{p.j.a.v.tilburg@tue.nl}
}
\def\titlerunning{Regular Expressions with Parallel Composition modulo Bisimilarity}
\def\authorrunning{Baeten, Luttik, Muller \& Van Tilburg}

\begin{document}
\maketitle

\begin{abstract}
  The languages accepted by finite automata are precisely the
  languages denoted by regular expressions. In contrast, finite
  automata may exhibit behaviours that cannot be described by regular
  expressions up to bisimilarity. In this paper, we consider
  extensions of the theory of regular expressions with various forms
  of parallel composition and study the effect on
  expressiveness. First we prove that adding pure interleaving to the
  theory of regular expressions strictly increases its expressiveness
  modulo bisimilarity. Then, we prove that replacing the operation for
  pure interleaving by ACP-style parallel composition gives a further
  increase in expressiveness. Finally, we prove that the theory of
  regular expressions with ACP-style parallel composition and
  encapsulation is expressive enough to express all \nfas{} modulo
  bisimilarity.  Our results extend the expressiveness results
  obtained by Bergstra, Bethke and Ponse for process algebras with
  (the binary variant of) Kleene's star operation.
\end{abstract}

\section{Introduction}

A well-known theorem by Kleene states that the languages accepted by
\nfas{} are precisely the languages denoted by a regular expression
(see, e.g.,~\cite{HMU06}). Milner, in~\cite{Mil84}, showed how regular
expressions can be used to describe \emph{behaviour} by defining an
interpretation of regular expressions directly as \nfas{}. He then
observed that the process-theoretic counterpart of Kleene's theorem
---stating that every \nfa{} is described by a regular expression---
fails: there exist \nfas{} whose behaviours cannot faithfully, i.e.,
up to bisimilarity, be described by regular expressions. Baeten,
Corradini and Grabmayer~\cite{BCG07} recently found a structural
property on \nfas{} that characterises those that are denoted with a
regular expression modulo bisimilarity. In this paper, we study to
what extent the expressiveness of regular expressions increases when
various forms of parallel composition are added.

Our first contribution, in Section~\ref{sec:bpa-pa}, is
to show that adding an operation for pure interleaving to regular
expressions strictly increases their expressiveness modulo
bisimilarity. A crucial step in our proof consists of characterising
the strongly connected components in \nfas{} denoted by regular
expressions. The characterisation allows us to prove a property
pertaining to the exit transitions from such strongly connected
components. If interleaving is added, then it is possible to denote
\nfas{} violating this property.

Our second contribution, in Section~\ref{sec:pa-acp},
is to show that replacing the operation for pure interleaving by
ACP-style parallel composition~\cite{BK84}, which implements a form of
synchronisation by communication between components, leads to a
further increase in expressiveness. To this end, we first characterise
the strongly connected components in \nfas{} denoted by regular
expressions with interleaving, and deduce a property on the
exit transitions from such strongly connected components. Then, we
present an expression in the theory of regular expressions with
ACP-style parallel composition that denotes a \nfa{} violating this
property.

Our third contribution, in Section~\ref{sec:LPS-ACP},
is to establish that adding ACP-style parallel composition and
encapsulation to the theory of regular expressions actually yields a
theory in which every \nfa{} can be expressed up to isomorphism, and
hence, since bisimilarity is coarser than isomorphism, also up to
bisimilarity. Every expression in the resulting theory, in turn,
denotes a \nfa{}, so this result can be thought of as an alternative
process-theoretic counterpart of Kleene's theorem.

The results in this paper are inspired by the results of Bergstra,
Bethke and Ponse pertaining to the relative expressiveness of process
algebras with a binary variant of Kleene's star operation. In
\cite{BBP94} they establish an expressiveness hierarchy on the
extensions of the process theories \BPA{}, \BPAd{}, \PA{}, \PAd{},
\ACP{}, and \ACPt{} with binary Kleene star. The reason that their
results are based on extensions with the binary version of the Kleene
star is that they want to avoid the process-theoretic complications
arising from the notion of intermediate termination (we say that a
state in a \nfa{} is intermediately terminating if it is terminating
but also admits a transition). Most of the expressiveness results
in~\cite{BBP94} are included in \cite{BFP01}, with more elaborate
proofs.

Casting our contributions mentioned above in process-theoretic
terminology, we establish a strict expressiveness hierarchy on the
process theories \BPAszo{} (regular expressions) modulo bisimilarity,
\PAszo{} (regular expressions with interleaving) modulo bisimilarity
and \ACPszo{} (regular expressions with ACP-style parallel composition
and encapsulation) modulo bisimilarity. The differences between the
process theories \BPAd{}, \PAd{} and \ACP{}
considered~\cite{BBP94,BFP01} and the process theories \BPAszo{},
\PAszo{} and \ACPszo{} considered in this paper are as follows: we
write $\dl$ for the constant deadlock which is denoted by $\delta$
in~\cite{BBP94,BFP01}, we include the unary Kleene star instead of its
binary variant, and we include a constant $\emp$ denoting the
successfully terminated process. The first difference is, of course,
cosmetic, and with the addition of the constant $\emp$ the unary and
binary variants of Kleene's star are interdefinable. So, our results
pertaining to the relative expressiveness of \BPAszo{}, \PAszo{} and
\ACPszo{} extend the expressiveness hierarchy of~\cite{BBP94,BFP01}
with the constant $\emp$.

In~\cite{BFP01} the expressiveness proofs are based on identifying
cycles and exit transitions from these cycles. There are two reasons
why the proofs in~\cite{BBP94} and~\cite{BFP01} cannot easily be
adapted to a setting with $\emp$.  First, in a setting with $\emp$ and
Kleene star there are cycles without any exit transitions. Second, the
inclusion of the empty process $\emp$ gives intermediate termination,
which, combined with the previously described different behaviour of
cycles, forces us to consider the more general structure of strongly
connected component.

\section{Preliminaries}\label{sec:prelim}

In this section, we present the relevant definitions for the process
theory \ACPszo{} and its subtheories \PAszo{} and \BPAszo{}. We give
their syntax and operational semantics, and the notion of (strong)
bisimilarity. We also introduce some auxiliary technical notions that
we need in the remainder of the paper, most notably that of strongly
connected component. The expressions of the process theory \BPAszo{}
are precisely the well-known regular expressions from the theory of
automata and formal languages, but we shall consider the automata
associated with them modulo bisimilarity instead of modulo language
equivalence.

The process theory $\ACPszo{}$ is parametrised by a non-empty set
$\Act$ of \emph{actions}, and a \emph{communication function} $\gamma$
on $\Act$, i.e., an associative and commutative binary partial
operation $\gamma: \Act\times\Act \rightharpoonup \Act$. $\ACPszo{}$
incorporates a form of synchronisation between the components of a
parallel composition by allowing certain actions to engage in a
\emph{communication} resulting in another action.  The communication
function $\gamma$ then defines which actions may communicate and what
is the result. The details of this feature will become clear when we
present the operational semantics of parallel composition.

The set of $\ACPszo$ expressions $\ACPszoExp$ is generated by
the following grammar:
\begin{equation*}
  p ::=\
    \dl\ \mid\
    \emp\  \mid\
    a\ \mid\
    p\seqc p\ \mid\
    p\altc p\ \mid\
    p\uks\ \mid\
    p\parc p\ \mid\
    \encap{H}{p}
\enskip,
\end{equation*}
with $a$ ranging over $\Act$ and $H$ ranging over subsets of $\Act$.

The process theory $\ACP{}$ (excluding the constants $\dl$ and $\emp$,
but including a constant $\delta$ with exactly the same behaviour as
$\dl$, and without the operation $\uks$) originates with~\cite{BK84}.
The extension of $\ACP{}$ with a constant $\emp$ was investigated by
\cite{KV85,BG87,Vra97} (in these articles, the constant was denoted
$\varepsilon$). The extension of $\ACP{}$ with the binary version of
the Kleene star was first proposed in~\cite{BBP94}.
  The reader already familiar with the process theory \ACPszo{} will
  have noticed that the operations $\lmerge$ (\emph{left merge}) and
  $\comm$ (\emph{communication merge}) are missing from our syntax
  definition. In \cite{BK84}, these operations are included as
  auxiliary operations necessary for a finite axiomatisation of the
  theory.  They do not, however, add expressiveness in our setting
  with Kleene star instead of a general form of recursion. We have
  omitted them to achieve a more efficient presentation of our
  results.

The constants $\dl$ and $\emp$ respectively stand for the deadlocked
process and the successfully terminated process, and the constants
$a\in\Act$ denote processes of which the only behaviour is to
execute the action $a$. An expression of the form $p\seqc q$ is called
a \emph{sequential composition}, an expression of the form $p\altc q$
is called an \emph{alternative composition}, and an expression of the
form $p\uks$ is called a \emph{star expression}. An expression of the
form $p\parc q$ is called a \emph{parallel composition}, and an
expression of the form $\encap{H}{p}$ is called an
\emph{encapsulation}.

From the names for the constructions in the syntax of \ACPszo{}, the
reader probably has already an intuitive understanding of the
behaviour of the corresponding processes. We proceed to formalise the
operational behaviour by means of a collection of operational rules
(see Table~\ref{tab:ACPszoSOS}) in the style of Plotkin's Structural
Operational Semantics~\cite{Plo04}.  Note how the communication
function in rule~\ref{osrule:merge3} is employed to model a form of
communication between parallel components: if one of the components of
a parallel composition can execute a transition labelled with $a$, the
other can execute a transition labelled with $b$, and the
communication function $\gamma$ is defined on $a$ and $b$, then the
parallel composition can execute a transition labelled with
$\gamma(a,b)$. (It may help to think of the action $a$ as standing for
the event of sending some datum $d$, the action $b$ as standing for
the event of receiving datum $d$, and the action $\gamma(a,b)$ as
standing for the event that two components communicate datum $d$.)
The $\Act$-labelled transition relation $\stepsym[\ACPszo]$ and the
termination relation $\termsym[\ACPszo]$ on $\ACPszoExp$ are the least
relations $\stepsym\subseteq\ACPszoExp\times\Act\times\ACPszoExp$ and
$\termsym\subseteq\ACPszoExp$ satisfying the rules in
Table~\ref{tab:ACPszoSOS}.

\begin{table}[htb]
  \hrule
  \begin{osrules}
   \osrule{}{\term{\emp}}\label{osrule:empterm} \quad
   \osrule{}{a \step{a} \emp}\label{osrule:astep} \quad
   \osrule{p \step{a} p'}{p \altc q \step{a} p'}\label{osrule:altcstepl} \quad
   \osrule{q \step{a} q'}{p \altc q \step{a} q'}\label{osrule:altcstepr} \quad
   \osrule{\term{p}}{\term{p \altc q}}\label{osrule:altcterml} \quad
   \osrule{\term{q}}{\term{p \altc q}}\label{osrule:altctermr}
  \\
   \osrule{p \step{a} p'}{p \seqc q \step{a} p' \seqc q}\label{osrule:seqcstepl} \quad
   \osrule{\term{p} & q \step{a} q'}{p \seqc q \step{a} q'}\label{osrule:seqcstepr} \quad
   \osrule{\term{p} & \term{q}}{\term{p \seqc q}}\label{osrule:seqcterm}\quad
   \osrule{p \step{a} p'}{p\uks \step{a} p' \seqc p\uks}\label{osrule:starstep} \quad
   \osrule{}{\term{p\uks}}\label{osrule:starterm}
  \\
   \osrule{p \step{a} p'}%
          {p \parc q \step{a} p' \parc q}\label{osrule:merge1} \quad
   \osrule{q \step{a} q'}%
          {p \parc q \step{a} p \parc q'}\label{osrule:merge2} \quad
   \osrule{\term{p} & \term{q}}%
          {\term{p \parc q}}\label{osrule:merge3}
  \\[1em]
   \osrule{p \step{a} p' & q \step{b} q' & \text{$\gamma(a,b)$ is defined}}
          {p \parc q \step{\gamma(a,b)} p' \parc q'}\label{osrule:comm} \quad
   \osrule{p \step{a} p' & a \not\in H}%
          {\encap{H}{p} \step{a} \encap{H}{p'}}\label{osrule:encap} \quad
   \osrule{\term{p}}{\term{\encap{H}{p}}}\label{osrule:encap2}
  \end{osrules}
  \hrule
  \caption{Operational rules for $\ACPszo$, with $a\in\Act$ and $H\subseteq\Act$.}\label{tab:ACPszoSOS}
\end{table}

The triple
  $\ACPszoTSS=(\ACPszoExp,\stepsym[\ACPszo],\termsym[\ACPszo])$,
consisting of the $\ACPszo$ expressions together with
the $\Act$-labelled transition relation and the termination predicate
associated with them, is an example of an $\Act$-labelled transition
system space. In general, an \emph{$\Act$-labelled transition system
  space} $(S,\stepsym,\termsym)$ consists of a (non-empty) set $S$, the
elements of which are called \emph{states}, together with an
$\Act$-labelled transition relation
  $\stepsym\subseteq S\times\Act\times S$
and a subset
  $\termsym\subseteq S$.
We shall in this paper consider two more examples of transition system
spaces, obtained by restricting the syntax of $\ACPszo$ and making
special assumptions about the communication function.

Next, we define the $\Act$-labelled transition system space
$\PAszoTSS=(\PAszoExp,\stepsym[\PAszo],\termsym[\PAszo])$
corresponding with the process theory $\PAszo$.  The set of $\PAszo$
expressions $\PAszoExp$ consists of the $\ACPszo$ process
expressions without occurrences of the construct $\encap{H}{\_}$.
The $\PAszo$ transition relation $\stepsym[\PAszo]$
on $\PAszoExp$ and the termination predicate $\termsym[\PAszo]$ on
$\PAszoExp$ are the transition relation and termination predicate
induced on $\PAszo$ expressions by the operational rules in
Table~\ref{tab:ACPszoSOS} minus the rules~\ref{osrule:comm}--%
\ref{osrule:encap2}. Alternatively (and equivalently) the transition
relation $\stepsym[\PAszo]$ can be defined as the restriction of the
transition relation $\stepsym[{\ACPszo[\Act,\emptyset]}]$, with
$\emptyset$ denoting the communication function that is everywhere
undefined, to $\PAszoExp$.

To define the $\Act$-labelled transition system space
  $\BPAszoTSS=(\BPAszoExp,\stepsym[\BPAszo],\termsym[\BPAszo])$
associated with the process theory $\BPAszo$, let $\BPAszoExp$ consist
of all $\PAszo$ expressions without occurrences of the construct
  ${\_}\parc{\_}$.
The $\BPAszo$ transition relation
$\stepsym[\BPAszo]$ and the $\BPAszo$ termination predicate
$\termsym[\BPAszo]$ are the transition relation and the termination
predicate induced on $\BPAszo$ expressions by the operational
rules in Table~\ref{tab:ACPszoSOS} minus the rules
\ref{osrule:merge1}--\ref{osrule:encap2}. That is, $\stepsym[\BPAszo]$
and $\termsym[\BPAszo]$ are obtained by restricting
$\stepsym[\ACPszo]$ and $\termsym[\ACPszo]$ to $\BPAszoExp$.

Henceforth, we shall omit the subscripts $\ACPszo$, $\PAszo$ and
$\BPAszo$ from transition relations and termination predicates
whenever it is clear from the context which transition relation or
termination predicate is meant. Furthermore, we shall often use
$\ACPszo$, $\PAszo$ and $\BPAszo$, respectively, to denote the
associated transition system spaces $\ACPszoTSS$, $\PAszoTSS$ and
$\BPAszoTSS$.

Let $\TSS=(S,\stepsym,\termsym)$ be an $\Act$-labelled transition
system space.  If $s,s'\in S$, then we write $s\step{}s'$ if there
exists $a\in\Act$ such that $s\step{a}s'$, and $s\notstep{}s'$ if
there exists no such $a\in\Act$. We denote by $\stepsym^{+}$ the
transitive closure of $\stepsym$, and by $\stepsym^{*}$ the
reflexive-transitive closure of $\stepsym$.  If $s\steprtc{} s'$ then
we say that $s'$ is \emph{reachable} from $s$; the set of all states
reachable from $s$ is denoted by $\Reach{s}$.  We say that a state $s$
is \emph{normed} if there exists $s'$ such that $s\steprtc{}s'$ and
$\term{s'}$.
$\TSS$ is called \emph{regular} if $\Reach{s}$ is finite for all $s\in
S$.

\begin{lemma}\label{lem:regularity}
  The transition system spaces \ACPszo{}, \PAszo{}, and \BPAszo{} are
  all regular.
\end{lemma}

With every state $s$ in $\TSS{}$ we can associate an
\emph{automaton}
(or: \emph{transition system})
  $(\Reach{s},
    \stepsym\cap(\Reach{s}\times\Act\times\Reach{s}),
    \termsym\cap\Reach{s},\
    s)$.
Its states are the states reachable from $s$, its transition relation
and termination predicate are obtained by restricting $\stepsym$ and
$\termsym$ accordingly, and the state $s$ is declared as the
\emph{initial state} of the automaton. If a transition system
space is regular, then the automaton associated with a state
in it is finite, i.e., it is a \nfa{} in the terminology of automata
theory. Thus, we get by Lemma~\ref{lem:regularity} that the operational
semantics of \ACPszo{}, and, a fortiori, that of \PAszo{} and \BPAszo{},
associates a \nfa{} with every process expression.

In automata theory, automata are usually considered as language
acceptors and two automata are deemed indistinguishable if they accept
the same languages. Language equivalence is, however, arguably too
coarse in process theory, where the prevalent notion is
bisimilarity \cite{Mil89,Par81}.

\begin{definition}\label{def:bisim}
  Let
    $\TSS[1]=(S_1,\stepsym[1],\termsym[1])$
  and
    $\TSS[2]=(S_2,\stepsym[2],\termsym[2])$
  be transition system spaces.
  A binary relation $\mathcal{R}\subseteq S_1\times S_2$ is a
  \emph{bisimulation} between $\TSS[1]$ and $\TSS[2]$ if it
  satisfies, for all $a\in\Act$ and for all $s_1\in S_1$ and $s_2\in S_2$
  such that $s_1\brel{}s_2$, the following conditions:
  \begin{enumerate}
   \renewcommand{\theenumi}{\roman{enumi}}
   \renewcommand{\labelenumi}{(\theenumi)}
   \item if there exists $s_1'\in S_1$ such that $s_1\step[1]{a}s_1'$, then there
     exists $s_2'\in S_2$ such that $s_2\step[2]{a}s_2'$ and
     $s_1'\mathrel{\mathcal{R}} s_2'$;
   \item if there exists $s_2'\in S_2$ such that $s_2\step[2]{a}s_2'$, then there
     exists $s_1'\in S_1$ such that $s_1\step[1]{a}s_1'$ and
     $s_1'\mathrel{\mathcal{R}} s_2'$; and
   \item $\term[1]{s_1}$ if, and only if, $\term[2]{s_2}$.
  \end{enumerate}
  States $s_1\in S_1$ and $s_2\in S_2$ are \emph{bisimilar} (notation:
  $s_1\bisim s_2$) if there exists a bisimulation $\brelsym$ between
  $\TSS[1]$ and $\TSS[2]$ such that $s_1\brel{}s_2$.
\end{definition}

To achieve a sufficient level of generality, we have defined
bisimilarity as a relation between transition system spaces; to obtain
a suitable notion of bisimulation between automata one should add the
requirement that the initial states of the automata be related.

Based on the associated transition system spaces, we can now define
what we mean when some transition system space is, modulo
bisimilarity, less expressive than some other transition system space.

\begin{definition}\label{def:expressive}
  Let $\TSS[1]$ and $\TSS[2]$ be transition system spaces.
  We say that $\TSS[1]$ is \emph{less expressive} than $\TSS[2]$
  (notation: $\TSS[1]\lexp\TSS[2]$) if every state in $\TSS[1]$
  is bisimilar to a state in $\TSS[2]$, and, moreover, there is a
  state in $\TSS[2]$ that is \emph{not} bisimilar to some state in
  $\TSS[1]$.
\end{definition}

When we investigate the expressiveness of \ACPszo{}, we want to
be able to choose $\gamma$. So, we are actually interested in the
expressiveness of the (disjoint) union of all transition system spaces
\ACPszo{} with $\gamma$ ranging over all communication functions.
We denote this transition system space by
  $\bigcup_\gamma\ACPszo$.
In this paper we shall then establish that
  $\BPAszo\lexp\PAszo\lexp\bigcup_\gamma\ACPszo$.

We recall below the notion of strongly connected component (see,
e.g.,~\cite{CLR01}) that will play an important r\^ole in
establishing that the above hierarchy of transition system spaces is
strict.

\begin{definition}\label{def:scc}
  A \emph{strongly connected component} in a transition system space
  $\TSS{}=(S,\stepsym,\termsym)$ is a maximal subset $C$ of $S$ such
  that $s \steprtc{} s'$ for all $s, s' \in C$.  A strongly connected
  component $C$ is \emph{trivial} if it consists of only one state,
  say $C=\{s\}$, and $s\notstep{} s$; otherwise, it is
  \emph{non-trivial}.
\end{definition}

Note that every element of a transition system space is an element of
precisely one strongly connected component of that space. Furthermore,
if $s$ is an element of a non-trivial strongly connected component,
then $s\step{}^{+}s$. Since in a strongly connected component from
every element every other element can be reached, we get as a
corollary to Lemma~\ref{lem:regularity} that strongly connected
components in \ACPszo{}, \PAszo{} and \BPAszo{} are finite.

Let $\TSS=(S,\stepsym,\termsym)$ be a transition system space, let
$s\in S$, and let $C\subseteq S$ be a strongly connected component in
$S$.  We say that $C$ is \emph{reachable} from $s$ if $s\steprtc{}s'$
for all $s'\in C$.

\begin{lemma}\label{lem:scc-bisim}
  Let $\TSS[1]=(S_1,\stepsym[1],\termsym[1])$ and
  $\TSS[2]=(S_2,\stepsym[2],\termsym[2])$ be regular transition
  system spaces, and let $s_1\in S_1$ and $s_2\in S_2$ be such that
    $s_1\bisim s_2$.
  If $s_1$ is an element of a strongly connected component $C_1$ in
  $\TSS[1]$, then there exists a strongly connected component $C_2$
  reachable from $s_2$ satisfying that
    for all $s_1'\in C_1$ there exists $s_2'\in C_2$ such that
      $s_1'\bisim s_2'$.
\end{lemma}
\section[Relative Expressiveness of BPA*01 and PA*01]{Relative Expressiveness of $\BPAszo$ and $\PAszo$}\label{sec:bpa-pa}

In~\cite{BBP94} it is proved that $\BPAsz$ is less expressive than
$\PAsz$. The proof in~\cite{BBP94} is by arguing that the $\PAsz$
expression
  $(a \seqc b)\bks c \parc d$
is not bisimilar with a \BPAsz{} expression.
(Actually, the \PAsz{} expression employed in~\cite{BFP01} uses only a
single action $a$, i.e., considers the \PAsz{} expression
  $(a \seqc a)\bks a \parc a$;
we use the actions $b$, $c$ and $d$ for clarity.)
An alternative and more general proof that the \PAsz{} expression
above is not expressible in \BPAsz{} is presented in
\cite{BFP01}. There it is established that the \PAsz{} expression above
fails the following general property, which is satisfied by all
\BPAsz{}-expressible automata:

\begin{quotation}\noindent
  If $C$ is a cycle in an automaton associated with a
  $\BPAsz$ expression, then there is at most one state $p \in C$
  that has an exit transition.
\end{quotation}
(A cycle is a sequence $(p_1,\dots,p_n)$ such that $p_i\step{}p_{i+1}$
($1\leq i <n$) and $p_n\step{}p_1$; an exit transition from $p_i$ is a
transition $p_i\step{}p_i'$ such that no element of the cycle is
reachable from $p_i'$.)


The following example shows that automata associated with
$\BPAszo$ expressions do not satisfy the property above.

\begin{example}
  Consider the automaton associated with the \BPAszo{} expression
    $\emp\seqc(a\seqc (a\altc \emp))\uks \seqc b$ (see Figure~\ref{fig:ex-cycle})
  with a cycle; both states on the cycle have a $b$-transition off the
  cycle.

  \begin{figure}[htb]
   \begin{center}
   \begin{transsys}(60,15)(0,2)
    \small
    \graphset{fangle=180}
    \node[Nmarks=i](p0)(0,12){$\emp\seqc(a\seqc (a \altc \emp))\uks \seqc b$}
    \node(p1)(60,12){$(a \altc \emp) \seqc (a \seqc (a \altc \emp))\uks \seqc b$}
    \node[Nmarks=f](p2)(0,0){$\emp$}
    \edge[curvedepth=3,sxo=10,exo=-12](p0,p1){$a$}
    \edge(p0,p2){$b$}
    \edge*[loopangle=-45,loopdiam=5](p1){$a$}
    \edge[curvedepth=3,sxo=-15,exo=10](p1,p0){$a$}
    \edge[curvedepth=4](p1,p2){$b$}
   \end{transsys}
   \end{center}
   \caption{A transition system in $\BPAszo$ with a cycle with multiple exit
     transitions.}\label{fig:ex-cycle}
  \end{figure}
\end{example}

In this section we shall establish that \BPAszo{} is less expressive
than \PAszo{}. As in~\cite{BFP01} we prove that \BPAszo{}-expressible
automata satisfy a general property that some
automaton expressible in \PAszo{} fails to
satisfy. We find it technically convenient, however, to base our
relative expressiveness proofs on the notion of strongly connected
component, instead of cycle. Note, e.g., that every process expression
is an element of precisely one strongly connected component, while it
may reside in more than one cycle. Furthermore, if $p \step{} q$ and
$p$ and $q$ are in distinct strongly connected components, then we can
be sure that $p\step{} q$ is an exit transition, while if $p$ and $q$
are on distinct cycles, then it may happen that $p$ is reachable from
$q$.


\subsection[Strongly Connected Components in BPA*01]{Strongly Connected Components in \BPAszo{}}

We shall now establish that a non-trivial strongly connected component
in \BPAszo{} is either of the form
  $\{p_1\seqc q\uks,\dots,p_n\seqc q\uks\}$
with $p_i$ ($0 \leq i \leq n$) reachable from $q$ and
  $\{p_1,\dots,p_n\}$ not a strongly connected component, or of the form
  $\{p_1\seqc q,\dots,p_n\seqc q\}$
where $\{p_1,\dots,p_n\}$ is a strongly connected component.  To this
end, let us first establish, by reasoning on the basis of the
operational semantics, that process expressions in a non-trivial
strongly connected component are necessarily sequential
compositions; at the heart of the argument will be the following measure
on process expressions.



\begin{definition}\label{def:counter-BPA}
  Let $p$ a $\BPAszo$ expression;
  then $\OC{p}$ is defined with recursion on the structure of $p$ by the
  following clauses:
  \begin{enumerate}
  \renewcommand{\theenumi}{\roman{enumi}}
  \renewcommand{\labelenumi}{(\theenumi)}
   \item\label{item:OCconst}
     $\OC{\dl}= \OC{\emp} = 0$, and $\OC{a} = 1$;
   \item\label{item:OCseqc}
     $\OC{p \seqc q} = 0$ if $q$ is a star expression,
       and $\OC{p \seqc q} = \OC{q} + 1$ otherwise;
   \item\label{item:OCaltc}
     $\OC{p + q} = \max{\OC{p},\OC{q}} + 1$;\ \text{and}
   \item\label{item:OCstar}
     $\OC{p\uks} = 1$.
  \end{enumerate}
\end{definition}

We establish that $\OC{\_}$ is non-increasing over transitions, and, in
fact, in most cases decreases.
\begin{lemma}\label{lem:OC-monotonic}
  If $p$ and $p'$ are $\BPAszo$ expressions such that
  $p \steptc{} p'$, then $\OC{p} \geq \OC{p'}$. Moreover, if
  $\OC{p}=\OC{p'}$, then $p=p_1\seqc q$ and $p'=p_1'\seqc q$ for some
  $p_1$, $p_1'$ and $q$.
\end{lemma}
\begin{proof}
  First, the special case of the lemma in which $p\step{}p'$ is
  established with induction on derivations according to the operational
  rules for \BPAszo{}. Then, the general case of the lemma follows from
  the special case with a straightforward induction on the length of a
  transition sequence from $p$ to $p'$.
\end{proof}

Let $P$ be a set of process expressions, and let $q$ be a process
expression; by $P\seqc q$ we denote the set of process expressions
$P\seqc q=\{p\seqc q\mid p\in P\}$.

\begin{lemma}\label{lem:sccseqc}
  If $C$ is a non-trivial strongly connected component in \BPAszo{},
  then there exist a set of process expressions $C'$ and a process
  expression $q$ such that $C=C'\seqc q$.
\end{lemma}

We proceed to give an inductive description of the non-trivial strongly
connected components in \BPAszo{}. The basis for the inductive
description is the following notion of basic strongly connected
component.

\begin{definition}\label{def:basic}
  A non-trivial strongly connected component $C=\{p_1,\dots,p_n\}$ in
  \BPAszo{} is \emph{basic} if there exist \BPAszo{} expressions
  $p_1',\dots,p_n'$ and a \BPAszo{} expression $q$ such that
    $p_i = p'_i \seqc q \uks$ $(1 \leq i \leq n)$
  and
    $\{p'_1, \dots , p'_n\}$
  is not a strongly connected component in \BPAszo{}.
\end{definition}

\begin{proposition}\label{prop:sccinduction}
  Let $C$ be a non-trivial strongly connected component in \BPAszo{}.
  Then either $C$ is basic, or there exist a non-trivial strongly
  connected component $C'$ and a \BPAszo{} expression $q$ such that
  $C=C'\seqc q$.
\end{proposition}
 \begin{proof}
  By Lemma~\ref{lem:sccseqc} there exists a set of states $C'$ and a
  \BPAszo{} expression $q$ such that $C = C' \seqc q$.
  If $C'$ is a non-trivial strongly connected component, then the
  proposition follows, so it remains to prove that if $C'$ is not a
  non-trivial strongly connected component, then $C$ is basic.
  Note that if $C'$ is not a strongly connected component, then there
  are $p, p' \in C'$ such that $p \notmstep{} p'$.
  Since $C$ is a non-trivial strongly connected component and
  $C = C' \seqc q$,  it holds that $p \seqc q \step{}^{+} p' \seqc q$.
  Using that $p\notmstep{} p'$, it can be established with
  induction on the length of the transition sequence from $p \seqc q$
  to $p' \seqc q$ that $q\step{}^{+} p'\seqc q$.
  It follows by \Lemma{OC-monotonic} that
    $\OC{q}\geq\OC{p'\seqc q}$,
  and therefore, according to the definition of $\OC{\_}$, $q$ must be
  a star expression. We conclude that $C$ is basic.
\end{proof}%

\subsection[BPA*01 is less expressive than PA*01]{$\BPAszo \lexp \PAszo$}

The crucial tool that will allow us to establish that \BPAszo{} is
less expressive than \PAszo{} will be a special property of states
with a transition out of their strongly connected component in
\BPAszo{}. Roughly, if $C$ is a strongly connected component in
\BPAszo{}, then all states with a transition out of $C$ have the same
transitions out of $C$.

\begin{definition}\label{def:sccexits}
  Let $C$ be a strongly connected component in the transition system
  space $\TSS{}=(S,\stepsym,\termsym)$ and let $s \in C$.
  An exit transition from $s$ is a pair $(a,s')$ such that
  $s\step{a} s'$ and $s' \not \in C$.
  We denote by $\Ext{s}$ the set of all \emph{exit transitions} from $s$,
  i.e.,
  $\Ext{s}=\{(a,s')\mid s \step{a} s' \land s' \not \in C \}$.
  An element $s \in C$ is called an \emph{exit state} if $\term{s}$
  or there exists an exit transition from $s$.
\end{definition}

\begin{example}
  Consider the automaton associated with the \BPAszo{}
  expression
    $\emp \seqc (a \seqc b \seqc (c \altc \emp))\uks \seqc d$,
  (see Figure~\ref{fig:ex-SCCs}).
  It has a strongly connecting component with two exit states,
  both with one exit transition $(d, \emp)$.

  \begin{figure}[htb]
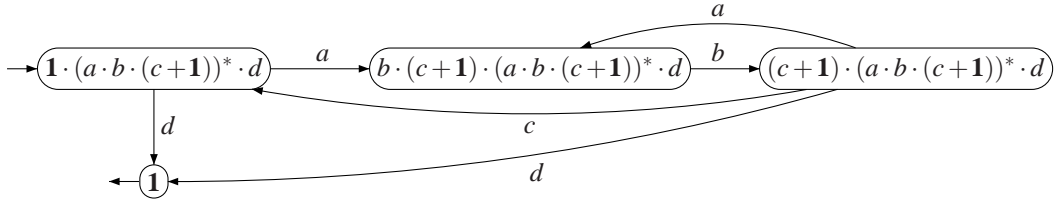

    \begin{center}
    \begin{transsys}(100,22)(0,5)
     \small
     \node[Nmarks=i](p1)(0,20){$\emp \seqc (a \seqc b \seqc (c \altc \emp))\uks \seqc d$}
     \node(p2)(50,20){$b \seqc (c \altc \emp) \seqc (a \seqc b \seqc (c \altc \emp))\uks \seqc d$}
     \node(p3)(100,20){$(c \altc \emp) \seqc (a \seqc b \seqc (c \altc \emp))\uks \seqc d$}
     \node[Nmarks=f,fangle=-180](p4)(0,5){$\emp$}
     \edge[ELpos=45](p1,p2){$a$}
     \edge(p2,p3){$b$}
     \edge[curvedepth=6](p3,p1){$c$}
     \edge[curvedepth=-6,ELside=r](p3,p2){$a$}
     \edge(p1,p4){$d$}
     \edge[curvedepth=4](p3,p4){$d$}
    \end{transsys}
    \end{center}
    \caption{A non-trivial strongly connected component in \BPAszo{} with multiple exit transitions.}\label{fig:ex-SCCs}
  \end{figure}
\end{example}

Non-trivial strongly connected components in \BPAszo{} arise from
executing the argument of a Kleene star. An exit state of a strongly
connected component in \BPAszo{} is then a state in which the
execution has the option to terminate. Due to the presence of $\dl$ in
\BPAszo{} this is, however, not the only type of exit state in
\BPAszo{} strongly connected components.

\begin{example}
  Consider the automaton associated with the \BPAszo{}
  expression
    $\emp \seqc (a \seqc ((b \seqc \dl) \altc \emp))\uks \seqc c$
  (see Figure~\ref{fig:ex-SCC2}).
  The strongly connected component contains two exit states and two
  (distinct) exit transitions.  One of these exit transitions leads to
  a deadlocked state.

  \begin{figure}[htb]
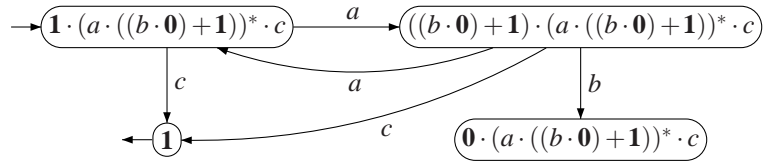

    \begin{center}
    \begin{transsys}(55,22)(0,5)
     \small
     \node[Nmarks=i](p1)(0,20){$\emp \seqc (a \seqc ((b \seqc \dl) \altc \emp))\uks \seqc c$}
     \node(p2)(55,20){$((b \seqc \dl) \altc \emp) \seqc (a \seqc ((b \seqc \dl) \altc \emp))\uks \seqc c$}
     \node(p3)(55,5){$\dl\seqc (a \seqc ((b \seqc \dl) \altc \emp))\uks \seqc c$}
     \node[Nmarks=f,fangle=-180](p4)(0,5){$\emp$}
     \edge[ELpos=45](p1,p2){$a$}
     \edge(p2,p3){$b$}
     \edge[curvedepth=4.5](p2,p4){$c$}
     \edge[curvedepth=6,sxo=-5](p2,p1){$a$}
     \edge(p1,p4){$c$}
    \end{transsys}    \end{center}
    \caption{A strongly connected component with normed exit transitions.}\label{fig:ex-SCC2}
  \end{figure}
\end{example}

The preceding example illustrates that the special property of
strongly connected components in \BPAszo{} that we are after, should
exclude from consideration any exit transition arising from an
occurrence of $\dl$. This is achieved in the following definitions.

\begin{definition}\label{def:live}
  Let $C$ be a strongly connected component and let $s \in C$.
  An exit transition $(a,s')$ from $s$ is \emph{normed} if $s'$ is normed.
  We denote by $\Extn{s}$ the set of normed exit transitions from $s$.

  An exit state $s\in C$ is \emph{alive} if $\term{s}$ or there exists
  a normed exit transition from $s$.
\end{definition}

\begin{lemma}\label{lem:starsteps}
  If $p\seqc q\uks\steprtc{}r$, then either
    there exists $p'$ such that $p\steprtc{}p'$ and $r=p'\seqc q\uks$
  or
    there exist $p'$ and $q'$ such that
      $p\steprtc{}p'$, $\term{p'}$,
      $q\steprtc{}q'$, and
      $r=q'\seqc q\uks$.
\end{lemma}

\begin{lemma}\label{lem:bsccnoextl}
  If $C$ is a basic strongly connected component,
  then $\Extn{p}=\emptyset$ for all $p\in C$.
\end{lemma}

\begin{lemma}\label{lem:seqcaes}
  Let $C$ be a non-trivial strongly connected component in
  \BPAszo{}, let $p\in C$, and let $q$ be a \BPAszo{} process
  expression such that $C \seqc q$ is a strongly connected component.
  Then $p\seqc q$ is an alive exit state in $C\seqc q$
    iff $p$ is an alive exit state in $C$ and $q$ is normed.
\end{lemma}

For a characterisation of the set of normed exit transitions of a
sequential composition, it is convenient to have the following
notation: if $E$ is a set of exit transitions $E$ and $p$ is a
\BPAszo{} expression, then $E\seqc p$ is defined by
  $E \seqc p = \{ (a,q \seqc p) \mid (a,q) \in E \}$.

\begin{lemma}\label{lem:seqcaes2}
  Let $C$ be a non-trivial strongly connected component in
  \BPAszo{}, let $p\in C$, and let $q$ be a normed \BPAszo{} process
  expression such that
    $C \seqc q$ is a strongly connected component.
  Then
  \begin{equation*}
    \Extn{p\seqc q}=
      \left\{\begin{array}{ll}
        \Extn{p}\seqc q
           \cup \{(a,r)\mid
                    r\not\in C\seqc q
                      \wedge
                    \text{$r$ is normed}
                      \wedge
                    q\step{a}r\}
          & \text{if $\term{p}$; and}\\
        \Extn{p}\seqc q
          & \text{if $\notterm{p}$.}
      \end{array}\right.
  \end{equation*}
\end{lemma}

\begin{proposition}\label{prop:aes-same-aet}
  Let $C$ be a non-trivial strongly connected component in \BPAszo{}.
  If $p_1$ and $p_2$ are alive exit states in $C$, then
    $\Extn{p_1}=\Extn{p_2}$.
\end{proposition}
\begin{proof}
  Suppose that $p_1$ and $p_2$ are alive exit states; we prove
  by induction on the structure of non-trivial strongly connected
  components in \BPAszo{} as given by
  Proposition~\ref{prop:sccinduction} that
    $\Extn{p_1} = \Extn{p_2}$
  and $\term{p_1}$ iff $\term{p_2}$.

  If $C$ is basic, then by Lemma~\ref{lem:bsccnoextl}
    $\Extn{p_1} = \emptyset = \Extn{p_2}$,
  and, since $p_1$ and $p_2$ are alive exit states, it also follows
  from this that both $\term{p_1}$ and $\term{p_2}$.

  Suppose that $C=C'\seqc q$, with $C'$ a non-trivial strongly
  connected component, and let $p_1',p_2'\in C'$ be such that
  $p_1=p_1'\seqc q$ and $p_2=p_2'\seqc q$. Since $p_1$ and $p_2$ are
  alive exit states, by Lemma~\ref{lem:seqcaes} so are $p_1'$ and
  $p_2'$. Hence, by the induction hypothesis,
    $\Extn{p_1'}=\Extn{p_2'}$
  and
    $\term{p_1'}$ iff $\term{p_2'}$.
  From the latter it follows that $\term{p_1}$ iff
  $\term{p_2}$.
  We now apply Lemma~\ref{lem:seqcaes2}: if, on the one hand,
  $\term{p_1}$ and $\term{p_2}$, then
  \begin{multline*}
    \Extn{p_1} =
      \Extn{p_1'}\seqc q
        \cup \{(a,r)\mid
                 r\not\in C
                   \wedge
                 \exists r'.\ q\step{a}r\steprtc{}\term{r'}\} \\
               =
      \Extn{p_2'}\seqc q
        \cup \{(a,r)\mid
                 r\not\in C
                   \wedge
                 \exists r'.\ q\step{a}r\steprtc{}\term{r'}\}
               =
      \Extn{p_2}
  \enskip,
  \end{multline*}
  and if, on the other hand, $\notterm{p_1}$ and $\notterm{p_2}$,
  then
    $\Extn{p_1}
      = \Extn{p_1'}\seqc q
      = \Extn{p_2'}\seqc q
      = \Extn{p_2}$.\qedhere
\end{proof}

\begin{figure}[htb]
 \begin{center}
  \begin{transsys}(15,18)(0,-2)
   \graphset{iangle=180,fangle=180}
   \node[Nmarks=i](p0)(0,15){$p_0$}
   \node(p1)(15,15){$p_1$}
   \node[Nmarks=f](p2)(0,0){$p_2$}
   \node(p3)(15,0){$p_3$}
   {\small
    {\graphset{curvedepth=2}
     \edge(p0,p1){$a$}\edge(p1,p0){$b$}
     \edge(p2,p3){$a$}\edge(p3,p2){$b$}
    }
    \edge(p0,p2){$c$}
    \edge(p1,p3){$c$}
   }
  \end{transsys}
  \caption{A \PAszo{}-expressible automaton that is not expressible in \BPAszo{}.}\label{fig:BPAszoprecPAszoexample}
 \end{center}
\end{figure}

  The \PAszo{} expression $p_0=\emp \seqc (a\seqc b)\uks \parc c$
  gives rise to the automaton shown in
  Figure~\ref{fig:BPAszoprecPAszoexample}. It has a strongly
  connected component $C=\{p_0,p_1\}$ of which the alive exit states
  have different normed exit transitions. Hence, by
  Proposition~\ref{prop:aes-same-aet}, $p_0$ is not
  \BPAszo{}-expressible.

\begin{theorem}\label{thm:bpa-pa}
  $\BPAszo$ is less expressive than $\PAszo$.
\end{theorem}

\section[Relative Expressiveness of PA*01 and ACP*01]{Relative Expressiveness of $\PAszo$ and $\ACPszo$}\label{sec:pa-acp}

The proof in~\cite{BFP01} that \PAsd{} is less expressive than
\ACPs{} uses the same expression as the one showing that \BPAsd{} is
less expressive than \PAsd{}, but it presupposes that $\gamma(c,d)=e$.
It is claimed that the associated automaton fails the
following general property of cycles in $\PAsd${}:
\begin{quotation}\noindent
  If $C$ is a cycle reachable from a $\PAsz$ process term and there is
  a state in $C$ with a transition to a terminating state, then all
  other states in $C$ have only successors in $C$.
\end{quotation}
The claim, however, is incorrect, as illustrated by the following
example.
(We present the example in the syntax of
$\PAszo$, but it has a straightforward translation into the syntax of
$\PAsd$.)

\begin{example}
  Consider the $\PAszo$ expression
    $(a\seqc (b \altc b \seqc b))\uks \seqc d$,
  from which the cycle
  \begin{equation*}
    C = \{\emp \seqc (a \seqc (b \altc b \seqc b))\uks \seqc d,\
          (b \altc b\seqc b)\seqc (a \seqc (b \altc b \seqc b))\uks \seqc d \}
  \end{equation*}
  is reachable.
  Clearly, the first expression in $C$ can perform a $d$-transition to
  $\emp$. Then, according to the property above, every other expression
  only has transitions to expressions in $C$. However,
  \begin{equation*}
    (b \altc b\seqc b)\seqc (a \seqc (b \altc b\seqc b))\uks \seqc d
       \step{b} b \seqc(a\seqc (b\altc b\seqc b))\uks \seqc d \not \in C
  \enskip.
  \end{equation*}
\end{example}

If we replace, in the property above, the notion of cycle by the
notion of strongly connected component, then the resulting property
does hold for \PAsz{}, but it still fails for \PAszo{}.

\begin{example}
  Consider the $\PAszo$ expression
    $(a\seqc b)\uks \parc c$%
  it gives rise to the following non-trivial strongly connected component:
   $\{ \emp \seqc (a\seqc b)\uks \parc c,\
       b \seqc (a\seqc b)\uks \parc c
    \}$.
  The expression $\emp \seqc (a\seqc b)\uks \parc c$ can do a
  $c$-transition to $\emp \seqc (a\seqc b) \uks \parc \emp$, for which
  the termination predicate holds, but at the same time
  $b \seqc (a\seqc b)\uks \parc c$ has an exit transition
    $(c, b \seqc (a\seqc b)\uks \parc \emp)$.
\end{example}

In this section we shall establish that \PAszo{} is less expressive
than \ACPszo{}. To this end, we apply the same method as in
Section~\ref{sec:bpa-pa}. First, we syntactically characterise the
non-trivial strongly connected components associated with \PAszo{}
expressions.  Then, we conclude that a weakened version of the
aforementioned property for strongly connected components holds in
\PAszo{}, and present an \ACPszo{} expression that does not satisfy
it.

\subsection[Strongly Connected Components in PA*01]{Strongly Connected Components in $\PAszo$}

To give a syntactic characterisation of the non-trivial strongly
connected components in \PAszo{}, we reason again about the
operational semantics. First, we extend the measure $\OC{\_}$ from
Section~\ref{sec:bpa-pa} to \PAszo{} expressions.

\begin{definition}\label{def:counter-PA}
  Let $p$ be a \PAszo{} expression; $\OC{p}$ is defined with recursion
  on the structure of $p$ by the clauses
    (\ref{item:OCconst})--(\ref{item:OCstar})
  in Definition~\ref{def:counter-BPA} with the following clause added:
  \begin{enumerate}
  \renewcommand{\theenumi}{\roman{enumi}}
  \renewcommand{\labelenumi}{(\theenumi)}
  \addtocounter{enumi}{4}
  \item
    $\OC{p \parc q} = 0$.
  \end{enumerate}
\end{definition}

With the extension, the non-increasing measure $\OC{\_}$ still in most cases
decreases over transitions.
\begin{lemma}\label{lem:OC-monotonic-PA}
  If $p$ and $p'$ are $\PAszo$ expressions such that $p \step{}^{+} p'$,
  then $\OC{p} \geq \OC{p'}$. Moreover, if $\OC{p}=\OC{p'}$, then either
    $p=p_1\seqc q$ and $p'=p_1'\seqc q$, or
    $p=p_1\parc p_2$ and $p'=p_1'\parc p_2'$
  for some process expressions $p_1$, $p_2$, $p_1'$, $p_2'$, and $q$.
\end{lemma}

\begin{lemma}\label{lem:parcsteps}
  Let $p$, $q$ and $r$ be \PAszo{} process expressions such that
  $p\parc q\steprtc{} r$.  Then there exist \PAszo{} process
  expressions $p'$ and $q'$ such that $r=p'\parc q'$, $p\steprtc{}
  p'$ and $q\steprtc{}q'$.
\end{lemma}

Let $P$ and $Q$ be sets of process expressions; by $P\parc Q$ we
denote the set of process expressions
  $P\parc Q=\{p\parc q\mid p\in P\wedge q\in Q\}$.
We also write $P \parc q$ and $p \parc Q$ for $P \parc \{q\}$
and $\{ p \} \parc Q$, respectively.

The proof of the following lemma, characterising the syntactic form of
non-trivial strongly connected components in \PAszo{}, is a
straightforward adaptation and extension of the proof of
Lemma~\ref{lem:sccseqc}, using \Lemma{OC-monotonic-PA} and
Lemma~\ref{lem:parcsteps} instead of \Lemma{OC-monotonic}.
\begin{lemma}\label{lem:sccseqcparc}
  If $C$ is a non-trivial strongly connected component in \PAszo{},
  then either there exist a set of process expressions $C'$ and a
  process expression $q$ such that $C=C'\seqc q$, or there exist
  strongly connected components $C_1$ and $C_2$ in \PAszo{}, at least
  one of them non-trivial, such that $C=C_1\parc C_2$.
\end{lemma}

The notion of \emph{basic} strongly connected component in \PAszo{} is
obtained from Definition~\ref{def:basic} by replacing \BPAszo{} by
\PAszo{} everywhere in the definition. In
Proposition~\ref{prop:sccinduction} we gave an inductive
characterisation of non-trivial strongly connected components in
\BPAszo{}. There is a similar inductive characterisation of
non-trivial strongly connected components in \PAszo{}, obtained by
adding a case for parallel composition.

\begin{proposition}\label{prop:sccinduction-PA}
  Let $C$ be a non-trivial strongly connected component in \PAszo{}.
  Then one of the following holds:
  \begin{enumerate}
  \renewcommand{\theenumi}{\roman{enumi}}
  \renewcommand{\labelenumi}{(\theenumi)}
  \item $C$ is a basic strongly connected component; or
  \item there exist a non-trivial strongly connected component $C'$ and
    a \PAszo{} expression $q$ such that $C=C'\seqc q$; or
  \item there exist strongly connected components $C_1$ and $C_2$, at
    least one of them non-trivial, such that $C=C_1\parc C_2$.
  \end{enumerate}
\end{proposition}

Note that, in the above proposition, one of the strongly connected
components $C_1$ and $C_2$ may be trivial in which case it
consists of a single \PAszo{} expression.

\subsection[PA*01 is less expressive than ACP*01]{$\PAszo \lexp \ACPszo$}

In Section~\ref{sec:bpa-pa} we deduced, from our syntactic
characterisation of strongly connected components in \BPAszo{}, the
property that all alive exit states of a strongly connected component
have the same sets of normed exit transitions. This property may fail
for strongly connected components in \PAszo{}: the automaton in
Figure~\ref{fig:BPAszoprecPAszoexample} is \PAszo-expressible, but the
alive exit states $p_0$ and $p_1$ of the strongly connected component
  $\{p_0, p_1\}$
have different normed exit transitions. Note, however, that these
normed exit transitions both end up in another strongly connected component
  $\{p_2, p_3\}$.
It turns out that we can relax the requirement on normed exit transitions
from strongly connected components in \BPAszo{} to get a requirement that
holds for strongly connected components in \PAszo{}. The idea is to
identify exit transitions if they have the same action and end up
in the same strongly connected component.

\begin{definition}\label{def:eqv-exit}
  Let $\TSS=(S,\stepsym,\termsym)$ be an $\Act$-labelled transition
  system space.
  We define a binary relation $\sim$ on $\Act\times S$
  by
    $(a,s) \sim (a',s')$
  iff
    $a=a'$
  and $s$ and $s'$ are in the same strongly connected component in
  $\TSS{}$.
\end{definition}

Since the relation of being in the same strongly connected component
is an equivalence on states in a transition system space, it is clear
that $\sim$ is an equivalence relation on exit transitions.
The following lemma will give some further properties of the relation
$\sim$ associated with \PAszo{}.

\begin{lemma}\label{lem:etcompatibility}
  Let $p$ and $q$ be \PAszo{} expressions, and let $a$ and $b$ be
  actions.
  If $(a,p)\sim(b,q)$, then
    $(a,p\seqc r)\sim(b,q\seqc r)$,
    $(a,p\parc r)\sim(b,q\parc r)$, and
    $(a,r\parc p)\sim(b,r\parc q)$.
\end{lemma}

To formulate a straightforward corollary of this lemma we use the following
notation: if $E$ is a set of exit transitions $E$ and $p$ is a
\PAszo{} expression, then $E\seqc p$, $E\parc p$ and $p\parc E$ are
defined by
\begin{gather*}
  E \parc p = \{ (a,q \parc p) \mid (a,q) \in E \}\enskip,\ \text{and}\ 
  p \parc E = \{ (a,p \parc q) \mid (a,q) \in E \}\enskip.
\end{gather*}

We are now in a position to establish a property of strongly connected
components in \PAszo{} that will allow us to prove that \PAszo{} is
less expressive than \ACPszo{}: a strongly connected component $C$ in
\PAszo{} always has a special exit state from which, up to $\sim$, all
exit transitions are enabled.


\begin{lemma}\label{lem:parcscc}
  Let $C_1$ and $C_2$ be sets of \PAszo{} expressions.  Then $C_1\parc
  C_2$ is a strongly connected component iff both $C_1$ and $C_2$ are
  strongly connected components. Moreover, $C_1\parc C_2$ is
  non-trivial iff at least one of $C_1$ and $C_2$ is non-trivial.
\end{lemma}

\begin{lemma}\label{lem:parcaes2-PA}
  Let $C_1$ and $C_2$ be a strongly connected components in \PAszo{},
  both with alive exit states.  Then $C_1\parc C_2$ is a strongly
  connected component with alive exit states too, and,
  for all $p \in C_1$ and $q \in C_2$,
  $\Extn{p \parc q} = (\Extn{p} \parc q) \cup (p \parc \Extn{q})$.
\end{lemma}

To formulate the special property of strongly connected components in
\PAszo{} that will allow us to prove that some \ACPszo{} expressions
do not have a counterpart in \PAszo{}, we need the notion of maximal
alive exit state.
\begin{definition}
  Let $\TSS{}=(S,\stepsym,\termsym)$ be an $\Act$-labelled transition
  system space, let ${\sim}\subseteq\Act\times S$ be the equivalence
  relation associated with $\TSS{}$ according to
  Definition~\ref{def:eqv-exit}, let $C$ be a strongly connected
  component in $\TSS{}$, and let $s\in C$ be an alive exit state.
  We say that $s$ is \emph{maximal} (modulo $\sim$) if for all alive exit
  states $s'\in C$ and for all $e'\in\Extn{s'}$ there exists an exit
  transition $e\in\Extn{s}$ such that $e\sim e'$.
\end{definition}

The following proposition establishes the property with which we shall
prove that \PAszo{} is less expressive than \ACPszo{}.

\begin{proposition}\label{prop:aes-same-taet}
  If $C$ is a strongly connected component in \PAszo{} and $C$ has an
  alive exit state, then $C$ has a maximal alive exit state.
\end{proposition}

\begin{figure}[htb]
 \begin{center}
  \begin{transsys}(30,17)
   \graphset{iangle=90,fangle=180}
   \node[Nmarks=i](p0)(15,15){$p_0$}
   \node(p1)(30,15){$p_1$}
   \node(p4)(0,15){$p_4$}
   \node(p2)(15,0){$p_2$}
   \node(p3)(30,0){$p_3$}
   \node[Nmarks=f](p5)(0,0){$p_5$}
   {\small
    {\graphset{curvedepth=2}
     \edge(p0,p1){$a$}\edge(p1,p0){$b$}
     \edge(p2,p3){$a$}\edge(p3,p2){$b$}
    }
    \edge(p0,p4){$d$}
    \edge(p2,p5){$d$}
    \edge(p0,p2){$c$}
    \edge(p1,p3){$c$}
    \edge(p4,p5){$c$}
    \edge[ELpos=33](p1,p2){$e$}
   }
  \end{transsys}
 \end{center}
 \caption{An \ACPszo{}-expressible automaton that is not expressible in \PAszo{}.}\label{fig:PAszoprecACPszoexample}
\end{figure}

Suppose $\gamma(b,c)=e$;
then the \ACPszo{} expression
  $p_0=\emp\seqc (a \seqc b)\uks \seqc d \parc c$
gives rise to the automaton shown in
Figure~\ref{fig:PAszoprecACPszoexample}. It has a strongly connected
component $C=\{p_0,p_1\}$, and none of its alive exit states is
maximal. Hence, by Proposition~\ref{prop:aes-same-taet}, $p_0$ is not
\PAszo{}-expressible.
\begin{theorem}\label{thm:pa-acp}
  $\PAszo$ is less expressive than $\bigcup_{\gamma}\ACPszo$.
\end{theorem}

\section[Every Finite Automaton is ACP*01-expressible]{Every Finite Automaton is  $\ACPszo$-expressible}\label{sec:LPS-ACP}

\newcommand{\enter}[1]{\ensuremath{\mathit{enter}_{#1}}}
\newcommand{\leave}[2]{\ensuremath{\mathit{leave}_{#1,#2}}}
\newcommand{\ActC}{\ensuremath{\mathalpha{C}}}

Milner observed in~\cite{Mil84} that there exist \nfas{} that are not
bisimilar to the \nfa{} associated with a \BPAszo{} expression. Our
proof of Theorem~\ref{thm:bpa-pa} has Milner's observation as an
immediate consequence: the \nfa{} associated with the \PAszo{}
expression used in the proof is not \BPAszo{}-expressible.  Similarly,
by Theorem~\ref{thm:pa-acp}, there are \nfas{} that are not
expressible in \PAszo{}.

In this section we shall prove that every \nfa{} \emph{is} expressible
in \ACPszo{}, for suitable choices of $\Act$ and $\gamma$, even up to
isomorphism. Before we formally prove the result, let us first
explain the idea informally, and illustrate it with an example.
The \ACPszo{} expression that we shall associate with a \nfa{} will
have one parallel component per state of the automaton,
representing the behaviour in that state (i.e., which outgoing
transitions it has to which other states and whether it is
terminating). At any time, one of those parallel components, the one
corresponding with the ``current state,'' has control. An
$a$-transition from that current state to a next state corresponds
with a communication between two components.  We make essential use of
\ACPszo{}'s facility to let the action $a$ be the result of
communication.

\begin{example}\label{ex:LPS}
  Consider the \nfa{} in Figure~\ref{fig:LPS}.

  \begin{figure}[htb]
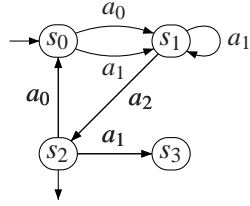

    \begin{center}
     \begin{transsys}(30,18)
      \graphset{iangle=180,fangle=-90}
      \node[Nmarks=i](p0)(0,15){$s_0$}
      \node(p1)(15,15){$s_1$}
      \node[Nmarks=f](p2)(0,0){$s_2$}
      \node(p3)(15,0){$s_3$}
      {\small
       \edge*[loopangle=0](p1){$a_1$}
       \edge[curvedepth=2](p0,p1){$a_0$}
       \edge[curvedepth=-2,ELside=r](p0,p1){$a_1$}
       \edge[ELpos=40](p1,p2){$a_2$}
       \edge(p2,p0){$a_0$}
       \edge(p2,p3){$a_1$}
       \edge[ELpos=40](p1,p2){$a_2$}
       \edge(p2,p0){$a_0$}
       \edge(p2,p3){$a_1$}
      }
     \end{transsys}
    \end{center}
   \caption{A finite automaton.}\label{fig:LPS}
  \end{figure}

  We associate with every state $s_i$ an \ACPszo{} expression
  $p_i$ as follows:
  \begin{equation*}
  \begin{array}{ll}
    p_0 =
       \Bigr(\enter{0}\seqc(\leave{0}{1}\altc\leave{1}{1})\Bigr)\uks
  \enskip,
&
    p_2 =
      \Bigr(\enter{2}\seqc(\leave{0}{0}\altc\leave{1}{3}\altc\emp)\Bigr)\uks
  \enskip,\\
    p_1 =
     \Bigr(\enter{1}\seqc a_1\uks\seqc(\leave{2}{2})\Bigr)\uks
  \enskip,
&
    p_3 =
      \Bigr(\enter{3}\seqc\dl\Bigr)\uks
  \enskip.
  \end{array}
  \end{equation*}
  Every $p_i$ has an $\enter{i}$ transition to gain control, and by
  executing a $\leave{k}{j}$ it may then release control to $p_j$
  with action $a_k$ as effect.
  We define the communication function so that an $\enter{i}$ action
  communicates with a $\leave{k}{i}$ action, resulting in the action $a_k$.
  Loops in the automaton (such as the loop on state $s_1$) require
  special treatment as they should not release control.

  Let $p_0'$ be the result of executing the $\enter{0}$-transition
  from $p_0$.
  We define the \ACPszo{} expression that simulates the \nfa{} in
  Figure~\ref{fig:LPS} as the parallel composition of $p_0'$, $p_1$,
  $p_2$ and $p_3$, encapsulating the control actions $\enter{i}$ and
  $\leave{k}{i}$, i.e., as
  \begin{equation*}
    \encap{\{\enter{i},\leave{k}{i}\mid 0\leq i \leq 3,\
                                        0\leq k \leq 2\}}{%
               p_0'\parc p_1\parc p_2 \parc p_3
             }
  \enskip.
  \end{equation*}
\end{example}

We now present the technique illustrated in the preceding example in
full generality.
Let
  $\FA{}=(S,\stepsym{},s_0,\termsym)$
be a \nfa{},
let
  $S=\{s_0,\dots,s_n\}$,
and let $A=\{a_1,\dots,a_m\}$ be the set of actions occurring on
transitions in $\FA{}$.
We shall associate with $\FA$ an \ACPszo{} expression $p_{\FA}$ that
has precisely one parallel component $p_i$ for every state $s_i$ in
$S$. To allow a parallel component to gain and release control,
we use a collection of \emph{control actions} $\ActC$, assumed to be
disjoint from $A$, and defined as
\begin{equation*}
  \ActC=\{\enter{i}\mid 1\leq i \leq n\}
      \cup
    \{\leave{k}{i}\mid 1\leq i \leq n,\ 1\leq k \leq m\}
\enskip.
\end{equation*}
Gaining and releasing control is modelled by the communication
function $\gamma$ satisfying:
\begin{equation*}
  \gamma(\enter{i},\leave{k}{j})=
    \left\{\begin{array}{ll}
      a_k & \text{if $i=j$; and}\\
      \text{undefined} & \text{otherwise.}
    \end{array}\right.
\end{equation*}
For the specification of the \ACPszo{} expressions $p_i$ we need one
more definition: for $1\leq i,j \leq n$ we denote by $K_{i,j}$ the set
of indices of actions occurring as the label on a transition from
$s_i$ to $s_j$, i.e.,
\begin{equation*}
  K_{i,j}=\{k\mid s_i\step{a_k}s_j\}
\enskip.
\end{equation*}
Now we can specify the \ACPszo{} expressions $p_i$ ($1\leq i \leq n$)
by
\begin{equation*}
  p_i = \emp \seqc
        \Bigr(\enter{i}
                \seqc
              (\sum_{k \in K_{i,i}} a_k)\uks
                \seqc
              (\sum_{\substack{1 \leq j \leq n \\ j \neq i}}\
           \sum_{k \in K_{i,j}}
             \leave{k}{i}\ {(+\;\emp)}_{\term{s_i}})\Bigr)\uks \;.
\end{equation*}
By ${(+\;\emp)}_{\term{s_i}}$ we mean that the summand $+\;\emp$ is
optional; it is only included if $\term{s_i}$.
The empty summation denotes $\dl$.
(We let $p_i$ start with $\emp$ to get that the \nfa{}
associated with $p_{\FA{}}$ is isomorphic and not just bisimilar with
$\FA{}$.)

Note that, in \ACPszo{}, every $p_i$ has a unique outgoing transition;
specifically $p_i \step{\enter{i}} p_i'$,
where $p_i'$ denotes:
\begin{equation*}
  p_i' = (\emp \seqc(\sum_{k \in K_{i,i}} a_k)\uks
                \seqc
              (\sum_{\substack{0 \leq j \leq n \\ j \neq i}}\
           \sum_{k \in K_{i,j}}
             \leave{k}{i}\ {(+\;\emp)}_{\term{s_i}}))\seqc p_i \;.
\end{equation*}

We now define
  $p_{\FA{}}=\encap{\ActC}{p_0'\merge p_1\merge \cdots\merge p_n}$.
Clearly, the construction of $p_{\FA{}}$ works for every \nfa{}
$\FA{}$. The bijection defined by
    $s_i \mapsto \encap{\ActC}{p_0\merge\cdots\merge p_{i-1}
                                  \merge p_i'\merge
                               p_{i+1}\merge\cdots\merge p_n}$
  is an isomorphism from $\FA{}$ to the automaton associated with
$p_{\FA}$ by the operational semantics.
We shall refer to $p_{\FA{}}$ as the \ACPszo{} expression associated
with $\FA{}$.

\begin{theorem}
  Let $\FA{}$ be a \nfa{}, and let $p_{\FA{}}$ be its
  associated \ACPszo{} expression. The automaton associated
  with $p_{\FA{}}$ by the operational rules for \ACPszo{} is isomorphic
  to $\FA{}$.
\end{theorem}

\begin{corollary}
  For every finite automaton $\FA{}$ there exists an instance
  of \ACPszo{} with a suitable finite set of actions \Act{} and a
  handshaking communication function $\gamma$ such that $\FA{}$ is
  \ACPszo{}-expressible up to isomorphism.
\end{corollary}

\section{Conclusion}\label{sec:conclusion}

In this paper we have investigated the effect on the expressiveness of
regular expressions modulo bisimilarity if different forms of parallel
composition are added.  We have established an expressiveness
hierarchy that can be briefly summarised as:
  $\BPAszo{}\lexp\PAszo{}\lexp\bigcup_{\gamma}\ACPszo{}$.
Furthermore, while not every \nfa{} can be expressed modulo
bisimilarity with a regular expression, it suffices to add a form of
\ACP{}-style parallel composition, with handshaking communication and
encapsulation, to get a language that is sufficiently expressive to
express all \nfas{} modulo bisimilarity.  This result should be
contrasted with the well-known result from automata theory that every
non-deterministic finite automaton can be expressed with a regular
expression modulo language equivalence.

As an important tool in our proof, we have characterised the strongly
connected components in \BPAszo{} and \PAszo{}. An interesting open
question is whether the two given characterisations are complete, in
the sense that a \nfa{} is expressible in \BPAszo{} or \PAszo{} iff
all its strongly connected components satisfy our
characterisation. If so, then our characterisation would
constitute a useful complement to the characterisation of \cite{BCG07}
and perhaps lead to a more efficient algorithm for deciding whether a
non-deterministic automaton is expressible.

In \cite{BFP01} it is proved that every finite transition system
without intermediate termination can be denoted in $\ACPszt$ up to
\emph{branching} bisimilarity \cite{GW96}, and that $\ACPsz$ modulo
(strong) bisimilarity is strictly less expressive than $\ACPszt$. In
contrast, we have established that every \nfa{} (i.e., every finite
transition system not excluding intermediate termination) is denoted
by an $\ACPszo$ expression. It follows that $\ACPszo$ and $\ACPszot$
are equally expressive.

An interesting question that remains is whether it is possible to omit
constructions from \ACPszo{} without losing expressiveness.
We conjecture that $\encap{H}{\_}$ cannot be omitted
without losing expressiveness: encapsulating $c$ in the \ACPszo{}
expression $\emp \seqc (a \seqc b)\uks \seqc b \parc c$, which is used in
Section~\ref{sec:pa-acp} to show that \PAszo{} is less expressive than
\ACPszo{}, yields a transition system that we think cannot be expressed
in \ACPszo{} without encapsulation.

\paragraph{Acknowledgement}
We thank Leonardo Vito and the other participants of the Formal
Methods seminar of 2007 for their contributions in an early stage of
the research for this paper.

\end{document}